\documentclass{article}

 \PassOptionsToPackage{numbers, compress}{natbib}
 \usepackage[final]{2017ARXIV}

\usepackage[utf8]{inputenc} % allow utf-8 input
\usepackage[T1]{fontenc}    % use 8-bit T1 fonts
\usepackage{hyperref}       % hyperlinks
\usepackage{url}            % simple URL typesetting
\usepackage{booktabs}       % professional-quality tables
\usepackage{amsfonts}       % blackboard math symbols
\usepackage{nicefrac}       % compact symbols for 1/2, etc.
\usepackage{microtype}      % microtypography
\usepackage{amsmath, amssymb}
\usepackage{graphics} % for pdf, bitmapped graphics files
\usepackage{graphicx}
\usepackage{epsfig} % for postscript graphics files
\usepackage{enumerate}
\usepackage{yhmath}
\usepackage{algorithm}
\usepackage[noend]{algpseudocode}
\newcommand{\N}{\mathbb{N}}

\newcommand{\R}{\mathbb{R}}

\newcommand{\Ac}{\mathcal{A}}
\newcommand{\Bc}{\mathcal{B}}

\newcommand{\Ec}{\mathcal{E}}
\newcommand{\Fc}{\mathcal{F}}

\newcommand{\Sc}{\mathcal{S}}

\newcommand{\Yc}{\mathcal{Y}}

\newcommand{\E}{\mathbb{E}}
\newcommand{\I}{\mathbb{I}}

\usepackage{color}

\usepackage{amsmath}
\usepackage{amssymb}
\usepackage{amsthm}

\usepackage{algorithm}
\usepackage[noend]{algpseudocode}

\usepackage{graphicx}
\usepackage{caption}
\usepackage{subcaption}
\usepackage{natbib}

\newtheorem{definition}{Definition}
\newtheorem{proposition}{Proposition}
\newtheorem{corollary}{Corollary}

\title{A Multi-Armed Bandit Approach for Online Expert Selection in Markov Decision Processes}

\author{
  Eric Mazumdar\\
  Department of Electrical Engineering and
  Computer Sciences \\
  University of California, Berkeley \\
  Berkeley, CA 94720\\
  \texttt{emazumdar@eecs.berkeley.edu} \\
  \And
  Roy Dong \\
  Department of Electrical Engineering and
  Computer Sciences \\
  University of California, Berkeley \\
  Berkeley, CA 94720\\
  \texttt{roydong@eecs.berkeley.edu}\\
  \And
  Vicen\c{c} R\'ubies Royo \\
  Department of Electrical Engineering and
  Computer Sciences \\
  University of California, Berkeley \\
  Berkeley, CA 94720\\
  \texttt{vrubies@berkeley.edu}\\
  \And
  Claire Tomlin \\
  Department of Electrical Engineering and
  Computer Sciences \\
  University of California, Berkeley \\
  Berkeley, CA 94720\\
  \texttt{tomlin@eecs.berkeley.edu}\\
  \And
  S. Shankar Sastry \\
  Department of Electrical Engineering and
  Computer Sciences \\
  University of California, Berkeley \\
  Berkeley, CA 94720\\
  \texttt{sastry@eecs.berkeley.edu}
}
\begin{document}

\maketitle

% make the title area

\begin{abstract}
    % !TEX root = 2017-arXiV.tex

We formulate a multi-armed bandit (MAB) approach to choosing expert policies online in Markov decision processes (MDPs). Given a set of expert policies trained on a state and action space, the goal is to maximize the cumulative reward of our agent. The hope is to quickly find the best expert in our set. The MAB formulation allows us to quantify the performance of an algorithm in terms of the regret incurred from not choosing the best expert from the beginning. We first develop the theoretical framework for MABs in MDPs, and then present a basic regret decomposition identity. We then adapt the classical Upper Confidence Bounds algorithm to the problem of choosing experts in MDPs and prove that the expected regret grows at worst at a logarithmic rate. Lastly, we validate the theory on a small MDP.

\end{abstract}

%%%%%%%%%%%%%%%%%%%%%%%%%%%%%%%%%%%%%%%%%%%%%%%%%%%%%%%%%%%%%%%%%%%%%%%%%%%%%%%%
\section{Introduction}
\label{sec:intro} 
% !TEX root = 2017-arXiV.tex

Artificial intelligence (AI) and machine learning (ML) algorithms are achieving great performance controlling systems too complicated to model. For example, deep Q-networks have been able to solve several Atari games~\citep{Minh_2013}. Similarly, deep reinforcement learning has also been able to tackle hard problems in robotics such as robot locomotion and grasping \citep{Schulman_15_1,Schulman_15_2,levine2016end}, tasks which are too complex to model due to the large number of degrees of freedom and/or the discontinuous dynamics that are present. 

%The application of deep neural networks to problems in reinforcement learning (RL) has had many well known and documented successes, not least of which is the success of deep Q-networks in solving Atari games~\citep{Minh_2013}. However, convolution neural networks have been shown to be particularly susceptible to perturbations in the input images~\citep{Dodge2016}. On top of this, the main goal of reinforcement learning has been largely to achieve expert performance on a given task, meaning that a certain amount of over-fitting to the task at hand occurs during training TODO FIND CITATION. These two phenomena means that the performance of learned policies are often not robust to changes in the underlying environment, which makes taking RL control policies from simulation to real life remains an unsolved problem.

Generally speaking, these systems achieve good performance when operating in conditions similar to the conditions of the training data. Our research presents methods for when the operating conditions are unknown and/or varying, and investigates guarantees on the performance of these methods.

In particular, we suppose that we have several batches of training data from several different operating conditions. We train a different learner for each operating condition: the trained learner will be referred to as an \emph{expert}. Intuitively, each expert is designed to perform well in one set of operating conditions.

In this situation, there are two natural questions that arise. First, if the operating conditions of the environment match the operating conditions of one batch of training data, can the controller identify this, and choose the appropriate expert? 
Second, if the environment presents operating conditions unlike those previously seen, can the controller choose an expert that still achieves the best performance possible?

One of the motivating examples for our work is the amazing performance that deep learning algorithms are achieving in the wild, as mentioned above.
Additionally, there is some promising research proving some of the performance properties of neural networks. For example, under some assumptions on the positive homogeneity of the component functions,~\citet{Haeffele2015} showed that all local optima are also global optima. As another example, researchers have recently shown that the structure of neural nets are expressive enough to show a large variety of functions~\citep{Raghu2017,Zhang2017}. 
However, all these results require that the training data is representative of the test data: if a dynamical system changes its operating conditions, there is little to ensure that performance will be maintained.

In this paper, we take a conjectural approach to proving good performance: supposing our learning algorithms do well in the operating conditions in which they were trained, how can we guarantee performance under hidden or changing operating conditions? 
In a sense, by integrating several experts into one controller and intelligently picking an expert should provide some form of robustness in some sense. In particular, we would like the systems we design in the future to be able to respond to sensor failures, random occlusions, and unexpected behavior, even if the system cannot directly detect these failure modes.

To such end, we present this preliminary work. We formulate a multi-armed bandit approach for online expert selection, where each expert provides a control policy for a Markov decision process (MDP). In contrast to classical multi-armed bandit problems, there is strong coupling between the states and the rewards of the Markov decision process, so one can not freely `switch' between experts without repercussions.

We provide an extension to the multi-armed bandit approach that handles this: at iteration $n$, we continue to use the same expert for $T_n$ time steps in the MDP. We are able to show that, under some ergodicity assumptions, the transient effects of previous experts die out for $T_n$ sufficiently large, and we can find a regret decomposition identity for the multi-armed bandit approach to MDPs, with an additional additive error term that is proportional to the mixing rates of the ergodic processes. As one example, we show that the Upper Confidence Bound algorithm can be used, and derive new regret bounds in the MDP setting.

In order to prove bounds the regret of our controller, we will have to assume that the operating conditions of the test data are fixed, and are represented in our training data. 
We hope to address these concerns in future work, but note that we empirically observe good performance in more general situations.

% !TEX root = 2017-arXiV.tex

\subsection{Notation}

For finite sets $\Ac$, we let $|\Ac|$ denote the cardinality. 
For any two sets $\Ac, \Bc$, we let $\Ac^\Bc$ denote the set of functions from $\Bc$ to $\Ac$. If $\Bc$ is finite, this can also be seen as a vector of size $|\Bc|$ with elements in $\Ac$.

For a set $\Ac$, we let $\Delta(\Ac)$ denote the set of probability distributions across $\Ac$. Throughout this paper, $\Ac$ will either be finite, in which case any subset will be considered measurable, or $\Ac \subset \R$, in which case the measurable sets will simply be the Borel subsets of $\Ac$.

If $X$ is a random element taking values in some set $\Ac$ and $\mu \in \Delta(\Ac)$, we will write $X \sim \mu$ to denote that $X$ is distributed according to $\mu$. Similarly, for random elements $X, Y$, we write $Y \vert X \sim P(X, \cdot)$ to state that the conditional distribution of $Y$ given $X$ follows a stochastic kernel $P$ evaluated at the realization of $X$. Alternatively, we can write this as $\Pr(Y \in A \vert X) = \int_A P(X,dy)$. 

For a random event $A$, we let $\I\{A\}$ indicate the random variable that equals 1 for all realizations in $A$ and 0 for all realizations not in $A$. Also, $\E X$ will denote the expected value of $X$, and $\E[X \vert Y]$ to be the conditional expectation of $X$ given $Y$, as defined in~\citep{Kallenberg2002}. We will also follow the common convention where the brackets of arguments in measures and stochastic kernels will be dropped when context is clear, e.g. $\Pr(X \geq t)$ instead of $\Pr(\{X \geq t\})$ and $P(s,a,s')$ instead of $P((s,a),\{s'\})$.

We let $\N$ denote the set of natural numbers $\{1,2,\dots\}$.

For a sequence of elements $a_0,a_1,\dots$, we will often write $a_{t_1:t_2}$ to denote $(a_{t_1},a_{t_1+1},\dots,a_{t_2})$. We will also generally use 0-indexing throughout this paper.

\section{Background}
\label{sec:background}

% !TEX root = 2017-arXiV.tex

The classic multi-armed bandit (MAB) problem was first described by~\citet{Robbins1952}, and provided the simplest framework to study the exploration vs. exploitation trade-off in sequential decision problems. Despite the apparent simplicity of the problem formulation, multi-armed bandit algorithms have proven to be extremely useful in practice~\citep{simon1989optimal, rothschild1974two, st2011online}, while often also providing strong theoretical guarantees of performance in terms of regret bounds~\citep{Hoeffding1963, Bubeck2012}.

In the context of sequential decision problems, another commonly used model is a MDP, where the classic MAB setting is a special case\footnote{The multi-armed bandit problem can be thought of as a one state MDP.}. Just as in the MAB formulation, the goal is to maximize the accumulated reward over some predefined number of steps. The main difference, however, stems from the fact that MDPs also includes coupling across time through its states and dynamics.

We've seen in recent work that MDPs have come into their own as models of complex systems, in particular those arising in AI and robotics \citep{Schulman_15_1,Schulman_15_2,Minh_2013} where states and dynamics need to be accounted for. A large body of literature has developed in learning good controllers for these MDP systems. 
Thus, perhaps counter-intuitively, in this work, we consider, a MAB approach to controlling MDPs. This is motivated by the fact that there is a large body of literature that can learn policies for control, each of which may be more or less effective in different operational conditions. 

Our work seeks to address the question: if we are given a set of a set of experts (or policies), can we decide at run-time which expert is the (overall) best fit to interact with the environment?

This approach relies on the fact that learning control policies for high-dimensional and complex systems is often very data-intensive, despite the large amount of recent work being conducted on speeding up the convergence of reinforcement learning algorithms ~\citep{Osband2016,Mnih2016}. Given this over-head for training new policies, we seek to make use of pre-trained policies to quickly achieve good, but perhaps not optimal performance.

The idea of expert selection for MDPs has been explored before. MABs are used to select the best quality grasp for robotic applications in~\citep{Laskey2015}. Similarly,~\citet{Matikainen} uses MABs to select among a large set of state machines to pick the best for a task involving maximum area coverage. In~\citep{Hamrick2017}, expert selection is used for meta-control in MDPs. 

The idea of expert selection in MDPs is also similar to the idea of restless bandits, where the underlying distribution of each arm evolves according to separate Markov processes that are independent of the pulls ~\citep{Bubeck2012}. Recent work has been conducted in deriving bounds on the regret of varying algorithms including the UCB algorithm in this setting ~\citep{Ortner2012,Tekin2012}. The critical difference in this case lies in the fact that in our setting, each pull of an arm (equivalently choice of an expert) effects the underlying dynamics which are also shared across all arms.

Thus, while MAB approaches for expert selection in MDPs tend to focus on isolated expert executions as the metric of performance for each bandit, in this work we define a time horizon $T$ of execution and equate the action of pulling the arm of one of the bandits, to executing a specific expert on the MDP for $T$ steps. In our work we devise a version of the UCB algorithm and show that this online expert selection paradigm, despite violating the assumption of independence between successive pulls, can be seen as an ``almost bandit" problem. We also provide regret bounds as well as some examples on simple MDPs.

%The idea of combining the MAB and MDP frameworks has not been widely explored. For instance, \cite{•} was the 

\section{Problem Formulation and Analysis}
\label{sec:problem}

% !TEX root = 2017-arXiV.tex

In this section, we formalize the problem under consideration. We first state the problem we are looking to address. We then describe the multi-armed bandit (MAB) formulation for Markov decision processes (MDPs), presenting a general regret decomposition identity for MAB algorithms in MDPs. Finally we extend the Upper Confidence Bounds (UCB) algorithm to the case of multi-armed bandits in Markov decision processes and show that the algorithm achieves logarithmically growing regret.

\subsection{Problem Statement} 

Our problem can be seen as the sequential interactions between a \emph{controller} and an \emph{environment}. In advance, both agents are given an \emph{observation} space $\Yc$ and an \emph{action} space $\Ac$. Upon initialization, the environment gives an observation $y_0 \in \Yc$ to the controller. Then, the controller picks an action $a_0 \in \Ac$ and issues this to the environment. The environment responds with an \emph{instantaneous reward} $r_0 \in \R$ and an observation $y_1 \in \Yc$. This repeats, with the controller choosing $a_t$ and the environment replying with $(r_{t},y_{t+1},)$. Throughout this paper we assume the rewards are bounded, and, without further loss of generality, suppose that $r_t \in [0,1]$ for all $t$.

The controller, in our formulation, receives suggested actions from a finite set of \emph{experts}. Let the set of experts be denoted $\Ec$, and let $|\Ec| = N$. At each $t$, the controller receives a suggested action from each expert; let the suggestion of expert $e \in \Ec$ at time $t$ be denoted $\hat a_t^e \in \Ac$, and the vector of suggested actions as $\hat a_t = (\hat a_t^e)_{e \in \Ec} \in \Ac^{\Ec}$.

At time $t$, the controller's action $a_t$ can depend on $(a_{0:t-1}, r_{0:t-1},y_{0:t}, \hat a_{0:t})$. Similarly, at time $t$, the environment's instantaneous reward and observation, $(r_t,y_{t+1})$, can depend on $(a_{0:t},y_{0:t},r_{0:t-1})$.\footnote{In control theory, this is referred to as a \emph{causal environment}, i.e. the environment's behavior at time $t$ cannot depend on the inputs given at future times.}

The goal of our controller is to achieve high \emph{cumulative rewards}, which shall be formally developed in the sequel. We shall provide an algorithm for picking $a_t$ to achieve such an end. The environment follows its own prerogative, whatever that may be. 

Note that, in this formulation, we suppose the controller has minimal information on the behavior of the environment. In fact, most of its information about what to do is mediated by the suggestions of the experts. Thus, rather than explicitly picking actions at each time step, we design a controller that instead decides which expert to listen to at each time step. 
As such, the controller chooses an expert $e_t$ and takes action $a_t = \hat a_t^e$.\footnote{In the sequel, we will note that it is important to track which expert told us to take $e_t$; in other words, if $e \neq e'$, then choosing $e_t = e, a_t = \hat a_t^e$ and choosing $e_t = e', a_t = \hat a_t^{e'}$ are different. This is the case even if $\hat a_t^e = \hat a_t^{e'}$.} 
We evaluate our controller's performance by considering how quickly it starts to listen to the expert who gives the best advice, in terms of rewards received, and the growth of the controller's \emph{regret}. Regret intuitively refers to the difference between the actual rewards received and the rewards the controller could have received, in retrospect. We will formally define regret in the analysis to follow.

% !TEX root = 2017-arXiV.tex

\subsection{Environment and Expert Assumptions}
We begin our formal construction of the multi-armed bandit in MDPs framework by stating our assumptions on both the environment and our set of experts.

\subsubsection{Assumptions on the environment}
For analysis, we will assume our environment is a finite \emph{Markov decision process (MDP)}. That is, our environment can be characterized by a finite set of \emph{states} $\Sc$. At time $0$, our environment is randomly distributed across states according to some initial probability distribution $\mu_0 \in \Delta(\Sc)$. Then, at each time step $t$, the environment transitions states according to some stochastic kernel $P$, such that $P(s,a,s')$ denotes the probability of being in state $s'$ at time $t+1$ given that the state at time $t$ is $s$ and the action taken by the controller at time $t$ is $a$. In other words, for each $s \in \Sc$ and $a \in \Ac$, we have that $P(s,a,\cdot) \in \Delta(\Sc)$. We let the sequence of realized states be denoted $(s_0,s_1,\dots)$, where $s_0 \sim \mu_0$, $s_1 \vert s_0, a_0 \sim P(s_0,a_0,\cdot)$, and so on, with the pattern $s_{t+1} \vert s_{0:t},a_{0:t} = s_{t+1} \vert s_{t},a_{t} \sim P(s_{t}, a_{t}, \cdot)$.

Further, we assume that the rewards are stochastic in nature and there is stochastic kernel $R$ such that $R(s,a,s',\cdot) \in \Delta([0,1])$ for every state $s \in \Sc$, action $a \in \Ac$, and next state $s' \in \Sc$. The rewards issued by the environment at each time $t$ are drawn according to this kernel: $r_t \vert s_{0:t+1},a_{0:t} = r_t \vert s_{t},a_{t},s_{t+1} \sim R(s_{t},a_{t},s_{t+1},\cdot)$. 

Lastly, this MDP is partially observed. There is a stochastic kernel $O$ such that $O(s,\cdot) \in \Delta(\Yc)$ for each state $s \in \Sc$. The observation is distributed as $y_t \vert s_{0:t} = y_t \vert s_t \sim O(s_t,\cdot)$.

\subsubsection{Assumptions on the experts}
We assume each expert $e \in \Ec$ has some mapping from the observation $y_t$ to a recommended action $\hat a_t^e$.\footnote{We note that this mapping can be random, so long as the stochasticity is independent of all the other processes thus mentioned.} 
We note that, by the structural assumptions on the MDP, this can be reduced to an the expert providing a stationary policy on the MDP. In particular, for expert $e \in \Ec$, we can let his stationary policy be denoted $\pi_e : \Sc \rightarrow \Delta(\Ac)$, which, sometimes, by mild abuse of notation, will also be thought of as $\pi_e : \Sc \times \Ac \rightarrow [0,1]$.\footnote{We can calculate this explicitly. If we let $f : \Yc \rightarrow \Ac$ denote the expert's mapping, then $\pi_e(s,a) = \int \I\{f(y) = a\} O(s,dy)$, which is the probability of taking actions given a state, averaged across the distribution across observations.} 

Note that each stationary policy induces a Markov chain.\footnote{We can simply take $\tilde P_e(s,s') = \sum_{a \in \Ac} P(s,a,s') \pi_e(s,a)$ to uncover the transition dynamics of this induced Markov chain.} We will assume that the Markov chain induced by each expert is aperiodic and irreducible, and therefore has a unique stationary distribution, as we will see in Proposition~\ref{prop:ergodic}.
%~\cite{Kallenberg2002}. We will let $\mu^e$ denote this stationary distribution.

\subsection{The Multi-Armed Bandit Approach to Choosing Experts in MDPs}

We now outline the general framework for choosing the best expert out of a set of experts in an MDP. We call this framework the multi-armed bandit approach to MDPs, and we summarize it in Algorithm \ref{alg:MABMDP}.\footnote{We quickly note that there are multiple time-scales occurring in this framework. There are the time steps of the MDP, which we will typically denote with $t$. There are also the iterations of the multi-armed bandit, which we will often index with $m$ or $n$. We write $t_n=\sum_{m=0}^{n-1} T_m$ to be the total time elapsed in the MDP by the $n$th iteration of the multi-armed bandit algorithm.}

First, the algorithm must pick an expert $e \in \Ec$. Let $t_1^e$ denote the time step of the MDP when expert $e$ is first chosen, $t_2^e$ the second, and so on for $t_k^e$. We will define $t_k^e = \infty$ if expert $e$ is not chosen a $k$th time.\footnote{It is interesting to note that $t_k^e$ is a random stopping time, whose realizations depend on the bandit algorithm in use. However, this will not affect the analysis.} When the expert is chosen at time $t_k^e$, we run this expert for $T^e_k$ time steps, incurring average cumulative reward $R_k^e = \frac{1}{T^e_k} \sum_{t = t_k^e}^{t_k^e+T^e_k-1} r_t$. Note that we choose the same expert for $T^e_k$ time steps, i.e. $e_t = e_{t_k^e}$ for $t \in \{t_k^e,t_k^e+1,\dots,t_k^e+T^e_k-1\}$. Once the chosen expert has been followed for $T^e_k$, we make another choice of expert, and run that expert for its specified time horizon. We then receive the new expert's average cumulative reward, and the algorithm repeats.

\begin{algorithm}[]
 % \algsetup{linenosize=\tiny}
  \caption{Multi-Armed Bandit Approach to MDPs}
  \begin{algorithmic}[1]
%  \algsetup{linenosize=\small}
  	\State \textbf{Input:} A set of experts $\Ec$, an environment modeled as an MDP, a sequence of times $\{T_n \in \N, n\ge0\}$
  	\State $t_0 \gets 0$
  	\State Initialize environment: $s_0 \sim \mu_0$
  	\State Receive observation $y_0$
 	\For{$n=0,1,2,...$}
		\State Make a choice of expert: $e_n \in \Ec$ 
		\For{$t=t_n$ to $t_{n}+T_n-1$}
			\State Receive vector of actions $\hat a_t$ from the experts
			\State Choose action $\hat a_t^{e_n}$ from $\hat a_t$ 
			\State Apply $\hat a_t^{e_n}$ to environment, receive reward and observation $(r_t,y_{t+1})$
			\State Give observation $y_{t+1}$ to the experts
		\EndFor
		\State $t_{n+1} \gets t_n+T_n$
 	\EndFor
	\end{algorithmic}
	\label{alg:MABMDP}
\end{algorithm}

Given the general framework we now show that under some mixing conditions for the Markov chains induced by each expert, this problem is, in some sense, `almost' a classical multi-armed bandit problem.  

Before we begin our theoretical treatment of the problem in earnest, we note that, for easy reference, we have summarized the most important notation in Figure~\ref{fig:notation_table}. 

First, we recall a convergence property for finite-state Markov chains:
\begin{proposition}[Ergodicity~\cite{Kallenberg2002}]
\label{prop:ergodic}
If the induced Markov chains for each expert $e \in \Ec$ are irreducible and aperiodic, then there exist a unique stationary distribution $\mu_e$. 

Furthermore, let $\tilde P_e$ denote the stochastic kernel of the induced Markov chain, and let $\tilde P_e^t$ denote the $t$-times composition of this stochastic kernel. Then, there exists constants $C_e, \alpha_e$ such that:
\begin{equation}
\max_{\mu \in \Delta(\Sc)} \left\| \sum_{s\in\Sc}\tilde  P_e^t(s,\cdot)\mu(s) - \mu_e(\cdot) \right\|_1 \leq C_e \alpha_e^t
\end{equation}
\end{proposition}

We begin our theoretical treatment by defining each expert's expected reward from its steady-state distribution $\mu_e$, $\bar R^e$:

\[
\bar R^e = \sum_{s \in \Sc} \sum_{a \in \Ac} \sum_{s' \in \Sc} \int r R(s,a,s',dr) P(s,a,s') \pi_e(s,a) \mu_e(s)\ \ \  \ \ \forall e \in \Ec
\] 
  
Having defined $\bar R^e$, we now give our definition of regret in the context of our problem:\\
\
\
\\

\begin{definition}[Regret definition for the multi-armed bandit approach to MDPs]
\label{def:regret}
We define the best expert as $e^* =\arg\max_{e \in \Ec} \bar R^e$, the expert which yields the highest expected average cumulative reward $\bar R^e$ from its steady-state distribution $\mu_e$. Finally we let $\bar R^* = \bar R^{e^*}$.

The cumulative regret, $r(n)$ after $n$ iterations of a multi-armed bandit algorithm is defined as:

 \[r(n)=n\bar R^* -  \sum_{k = 0}^{n}\frac{1}{T_k}\sum_{t=t_k}^{t_{k+1}-1} r_t\]
\end{definition}

%%%%%%%%%%%%%%%%%%%%%%%%%%%%%%%%%%%%%%%%%%%%%%%%%%%%%%%%%%%%%%%%%%%%%%%%%%%%%%%%%%%%%%%%%%%%%%%%%%%%%%%%%%%%%%%%%%%%%%%%%%%%%%%%%%%%%%%%%%%%  Table

\begin{figure}[!ht]
\begin{center}

\begin{tabular}{|| p{1in} | p{4in} ||}
\hline
Variable 		& Interpretation \\
\hline
$y_t$ 		& The observation given by the environment at time $t$, taking values in $\Yc$. \\
$a_t$ 		& The action chosen by the controller at time $t$, taking values in $\Ac$. \\
$r_t$ 		& The reward given by the environment at time $t$, taking values in $[0,1]$. \\
$\Ec$ 		& The set of experts, also often referred to as bandits. \\
$\hat a_t^e$	& The advice of expert $e \in \Ec$ at time $t$. \\
$e_t$ 		& The expert chosen by the controller at time $t$. \\
$P(s,a,\cdot)$ 	& The stochastic kernel dictating the transition probabilities of our MDP. \\
$R(s,a,s',\cdot)$ 	& The stochastic kernel dictating the distribution across rewards. \\
$O(s,\cdot)$ 	& The stochastic kernel dictating the distribution across observations. \\
$\pi_e$ 		& The stationary policy induced by expert $e \in \Ec$. \\
$\mu_e$ 		& The stationary distribution of the Markov chain induced by $\pi_e$. \\
$T_n$ 			& The number of MDP time steps in the expert chosen at the $n$th iteration of the multi-armed bandit algorithm is used for. \\
$t_n$           & The time step of the MDP when the algorithm makes its $n$th choice of expert to follow. $t_n=\sum_{m=0}^{n-1}T_m$ \\
$t_k^e$ 		& The time step of the MDP when expert $e \in \Ec$ is chosen for the $k$th time. \\
$T_k^e$         & The number of MDP time steps expert $e$ is used for when it is called for the $k$th time\\
$R_k^e$ 		& The average cumulative reward from expert $e \in \Ec$ the $k$th time it is chosen, given by $\frac{1}{T} \sum_{t = t_k^e}^{t_k^e+T^e_k-1} r_t$. \\
$\Fc_k^e$ 	& The filtration of our constructed martingale for expert $e \in \Ec$. \\
$S_k^e$ 		& The martingale for expert $e \in \Ec$. \\
$C_e, \alpha_e$ 	& The mixing constants of the Markov chain induced by expert $e$. \\
$\bar R^e$ 	& The expected steady-state reward for expert $e$. \\
$e^*, \bar R^*$ 	& The best expert and the expected reward of the stationary distribution of the best expert. \\
$C_*, \alpha_*$ & The mixing constants for the best expert $e^*$. \\
$K_e$ 		& A placeholder term, defined as $K_e = \frac{C_e}{1 - \alpha_e}$. \\
$T_e(n)$		& The number of times expert $e \in \Ec$ has been called by the $n$th iteration of the multi-armed bandit algorithm. \\
$\Delta_e$ 	& The expected regret of expert $e$, given by $\Delta_e = \bar R^* - \bar R^e$. \\
\hline
\end{tabular}

\end{center}
\caption{A table outlining the notation used throughout this paper.}
\label{fig:notation_table}
\end{figure}

%%%%%%%%%%%%%%%%%%%%%%%%%%%%%%%%%%%%%%%%%%%%%%%%%%%%%%%%%%%%%%%%%%%%%%%%%%%%%%%%%%%%%%%%%%%%%%%%%%%%%%%%%%%%%%%%%%%%%%%%%%%%%%%%%%%%%%%%%%%% Convergence property of Markov Chains

Having defined the important terms, we now present the necessary building blocks for a general regret decomposition for the multi-armed bandit in MDPs problem. 
In essence, we know that if our controller exclusively listened to expert $e \in \Ec$, then the MDP would converge to the stationary distribution $\mu_e$, and earn an expected steady state reward of $\bar R^e$ . Corollary ~\ref{corr:mc_bnd} tells us that, for a large enough time horizon, $T_n$, our observed $R_k^e$ is close to $\bar R^e$.\\

%%%%%%%%%%%%%%%%%%%%%%%%%%%%%%%%%%%%%%%%%%%%%%%%%%%%%%%%%%%%%%%%%%%%%%%%%%%%%%%%%%%%%%%%%%%%%%%%%%%%%%%%%%%%%%%%%%%%%%%%%%%%%%%%%%%%%%%%%%%% Convergence of Rewards to stationary rewards

\begin{corollary}
\label{corr:mc_bnd}
Suppose the induced Markov chains for each expert $e \in \Ec$ are irreducible and aperiodic. Pick any $n$, and let $e = e_{t_n}$ be the expert chosen at the $n$th iteration. We have that:
\[
\left| \bar R^{e} - \E\left[ \frac{1}{T_n}\sum_{t = t_n}^{t_{n+1}-1}r_t \middle| s_{t_n},e \right] \right| \leq \frac{C_e}{T_n} \frac{1 - \alpha_e^{T_n}}{1 - \alpha_e} \text{ almost surely}
\]

We supply the proof in Appendix \ref{sec:cor1}.

\end{corollary}

Now, let $T_e(n) = \sum_{m = 0}^{n-1} \I\{e_{t_m} = e\}$, the number of times the expert $e$ has been chosen up to iteration $n$. Note that this is a random quantity. Additionally, let us define the expected regret of expert $e$ as $\Delta_e = \bar R^* - \bar R^e$. Note that $\Delta_{e^*} = 0$. Finally, we let the time horizon $T_n$, that a chosen expert is played for, to vary over the course of the algorithm. We further require that $T_n\ge T_0$ for all $n$.\\

We now state our main proposition on the expected regret of a multi-armed bandit algorithm applied to an MDP.\\

%%%%%%%%%%%%%%%%%%%%%%%%%%%%%%%%%%%%%%%%%%%%%%%%%%%%%%%%%%%%%%%%%%%%%%%%%%%%%%%%%%%%%%%%%%%%%%%%%%%%%%%%%%%%%%%%%%%%%%%%%%%%%%%%%%%%%%%%%%%% Regret Decomposition

\begin{proposition}[Regret decomposition identity]
\label{prop:rdi}
If the induced Markov chains for each expert $e \in \Ec$ are irreducible and aperiodic, then 
the expected cumulative regret at time $n$ can be bounded with:
\begin{equation}
\E[r(n)]\le \sum_{\substack{e \in \Ec \\ e\ne e^*}} \E[ T_e(n) ] \left[ \Delta_e + \frac{C_e}{T_0(1 -\alpha_e)} \right] + \frac{C_{*}}{1-\alpha_{*}}\sum_{k = 0}^{n-1} \frac{1}{T_k}
\end{equation}
Here, $C_*, \alpha_*$ are the mixing constants for the best expert $e^*$.
\end{proposition}

\begin{proof}

Note that $\sum_{e \in \Ec} T_e(n) = n$ and $\sum_{m = 0}^{n-1} \I\{e_{mT} = e\} = T_e(n)$, both almost surely. 
{\allowdisplaybreaks
\begin{align*}
r(n)&=n\bar R^* -\sum_{m = 0}^{n-1}\sum_{t = t_m}^{t_{m+1}-1} \frac{r_t}{T_m}\\
&=n \bar R^* - \sum_{e \in \Ec} T_e(n) \bar R^e + \sum_{e \in \Ec} T_e(n) \bar R^e -\sum_{m = 0}^{n-1}\sum_{t = t_m}^{t_{m+1}-1} \frac{r_t}{T_m}\\
&=\sum_{e \in \Ec} T_e(n) [\bar R^* - \bar R^e] + \sum_{e \in \Ec} T_e(n) \bar R^e - \sum_{m = 0}^{n-1} \sum_{e \in \Ec} \I\{e_{t_m} = e\} \sum_{t = t_m}^{t_{m+1}-1} \frac{r_t}{T_m}\\
& =\sum_{e \in \Ec} T_e(n) [\bar R^* - \bar R^e] + \sum_{m = 0}^{n-1} \sum_{e \in \Ec} \I\{e_{t_m} = e\} \left[ \bar R^{e_{t_m}} -\sum_{t = t_m}^{t_{m+1}-1} \frac{r_t}{T_m} \right] \\
&=\sum_{e \in \Ec} T_e(n) \Delta_e +\sum_{m = 0}^{n-1} \sum_{e \in \Ec} \I\{e_{t_m} = e\} \left[ \bar R^{e_{t_m}} -\sum_{t = t_m}^{t_{m+1}-1} \frac{r_t}{T_m} \right]
\end{align*}
}

Thus we get that a first decomposition of the regret is given by:
\begin{align}
r(n)&=\sum_{e \in \Ec} T_e(n) \Delta_e +...\label{eq:Regret1}\\
&\quad+\sum_{m = 0}^{n-1} \sum_{\substack{e \in \Ec \\ e\ne e^*}} \I\{e_{t_m} = e\} \left[ \bar R^{e_{t_m}} -\sum_{t = t_m}^{t_{m+1}-1} \frac{r_t}{T_m} \right]+... \label{eq:Regret2}\\
&\quad+ \sum_{k = 0}^{n-1}\I\{e_{t_k} = e^*\}\left[ \bar R^{*} -\sum_{t = t_m}^{t_{m+1}-1} \frac{r_t}{T_m} \right]
\label{eq:Regret3}
\end{align}

We note that the first term describes the steady state regret of all the suboptimal agents. The second term describes the transients coming from the fact that each suboptimal expert is not truly operating in their steady state. The third term describes the transient regret coming from choosing the optimal expert.\\

We first find an upper bound for the expectation of the second term, (\ref{eq:Regret2}). By utilizing the tower property of the conditional expectation, noting the appropriate measurability, and invoking Corollary~\ref{corr:mc_bnd}, we get that:
\[
\E \left[ \sum_{m = 0}^{n-1} \sum_{\substack{e \in \Ec \\ e\ne e^*}} \I\{e_{t_m} = e\} \left[ R^{e_{t_m}} -\sum_{t = t_m}^{t_{m+1}-1} \frac{r_t}{T_m} \right] \right] 
\]
\begin{align*}
 &= \E \left[\sum_{m = 0}^{n-1} \sum_{\substack{e \in \Ec \\ e\ne e^*}}\I\{e_{t_m} = e\} \E \left[  \left( \bar R^{e_{t_m}} - \sum_{t = t_m}^{t_{m+1}-1} \frac{r_t}{T_m} \right) \middle| s_{t_m},e_{t_m} \right] \right]\\
 & \leq \E \left[\sum_{m = 0}^{n-1} \sum_{\substack{e \in \Ec \\ e\ne e^*}} \I\{e_{t_k} = e\} \E \left[  \left( \frac{C_{e_{t_m}}(1-\alpha^{T_m}_{e_{t_m}})}{T_m (1-\alpha_{e_{t_m}})}\right) \middle| s_{t_m},e_{t_m} \right] \right]\\
 &\leq \E \left[\sum_{m = 0}^{n-1} \sum_{\substack{e \in \Ec \\ e\ne e^*}}\I\{e_{t_m} = e\} \left( \frac{C_{e_{t_m}}}{T_0 (1-\alpha_{e_{t_m}})}\right) \right] \\
 &= \E \left[\sum_{m = 0}^{n-1} \sum_{\substack{e \in \Ec \\ e\ne e^*}}\I\{e_{t_m} = e\} \left( \frac{C_{e}}{T_0 (1-\alpha_{e})}\right) \right] \\
 &=\sum_{\substack{e \in \Ec \\ e\ne e^*}}\E \left[T_e(n)\right] \left( \frac{C_{e}}{T_0 (1-\alpha_{e})}\right)  \\
\end{align*}

The second line comes from invoking Corollary~\ref{corr:mc_bnd} and the third from the fact that $T_n\ge T_0$ and $\alpha_e<1$ for all $e\in \Ec$.\\

We now find an upper bound on the expectation of the third term in (\ref{eq:Regret1}).

\begin{align*}
\E \left[ \sum_{m = 0}^{n-1}\I\{e_{t_m} = e^*\}\left[ \bar R^* -\sum_{t = t_m}^{t_{m+1}-1} \frac{r_t}{T_m} \right] \right] &=
\E \left[ \sum_{m = 0}^{n-1}\I\{e_{t_m} = e^*\} \E \left[ \bar R^* -\sum_{t = t_m}^{t_{m+1}-1} \frac{r_t}{T_m} \middle| s_{t_m},e^* \right] \right]\\
&\le \E \left[ \sum_{m = 0}^{n-1}\I\{e_{t_m} = e^*\} \E \left[ \frac{C_*}{T_m (1-\alpha_*)}  \middle| s_{t_m},e^* \right] \right]\\
&= \E \left[ \sum_{m = 0}^{n-1}\I\{e_{t_m} = e^*\} \frac{C_*}{T_m (1-\alpha_*)} \right] \\
&\le \sum_{m = 0}^{n-1}\frac{C_*}{T_m (1-\alpha_*)}\\
&= \frac{C_*}{1-\alpha_*}\sum_{m = 0}^{n-1} \frac{1}{T_m}\\
\end{align*}

Now putting all of this together, and setting $K_e=\frac{C_e}{1 -\alpha_e}$ for each expert $e$ in $\Ec$ we get that the expected regret, $\E[r(n)]$, satisfies:

\[
\E[r(n)]\le \sum_{\substack{e \in \Ec \\ e\ne e^*}} \E[ T_e(n) ] \left[ \Delta_e + \frac{K_e}{T_0} \right] + K_*\sum_{m = 0}^{n-1} \frac{1}{T_m}
\]
\end{proof}

We note that this regret decomposition is similar to the case for the classical MAB formulation. Indeed, we note that in the limit where $T_0\rightarrow \infty$, this decomposition is exactly the same as the classical case.  The term $\Delta_e + \frac{K_e}{T_0}$ is the regret coming from choosing a sub-optimal expert. The term $K_*\sum_{m = 0}^{n-1} \frac{1}{T_m}$ comes from the fact that, because of the dynamics of the MDP, choosing the optimal expert will not immediately give the steady-state reward. We note that the sequence $\{T_n,n\ge0\}$ can be chosen strategically to control how much this term contributes to the overall regret.

Thus, we have shown how the problem of choosing the best expert out of a set in MDPs can be recast as a multi-armed bandit problem. We have also outlined a general framework for using MAB algorithms to select control policies for dynamical systems, despite the coupling between sequential actions and issues with causality. What is important is that the algorithms are robust in the following sense: small perturbations to the distribution for each sample will not affect the convergence properties of the algorithm, even when these perturbations do not necessarily follow any sort of structure. We show how the classical UCB algorithm can be extended to the experts in MDPs problem in the sequel.

% !TEX root = 2017-arXiV.tex

\subsection{Upper Confidence Bound Algorithm}

As an example of the regret analysis we mentioned above, we consider the Upper Confidence Bound (UCB) algorithm adapted to MDPs. Like the classical UCB algorithm introduced by~\citep{lai1985}, the crux of the algorithm lies in the construction of the confidence bounds, $c^e_{k,n}$. For each $e \in \Ec$ and $k, n$, these confidence bounds are independent of the realized values of the random variable, and can be computed a priori.  We first discuss the construction of the confidence bounds, and then develop bounds on the expected regret of the algorithm.

\begin{algorithm}[]
 % \algsetup{linenosize=\tiny}
  \caption{Upper Confidence Bounds for MDPs}
  \begin{algorithmic}[1]
%  \algsetup{linenosize=\small}
  	\State \textbf{Input:} A set of experts $\Ec$, a sequence of times $\{T_n \in \N, n\ge0\}$, confidence bounds $c^e_{k,n}$ 
  	\State $t_0 \gets 0$
  	\State $k_e \gets 0$ for $e \in \Ec$ \{the number of time expert $e$ has been chosen\}
% 	\State $c^e_{0,0} \gets \infty \ \ \forall \ i \in \{1,...N\}$
% 	\State $R_e \gets 0 \ \ \forall \ e \in \Ec$ \{the cumulative average reward of expert $i$\}
 	\State $S_e \gets 0 \ \ \forall \ e \in \Ec$ \{the cumulative per-MDP-time-step reward of expert $i$\}
 	\For{$n=0,1,2,...$}
		\State $e_{t_n} \gets \arg\max_{e \in \Ec} \{S_e/k_e+c^e_{k_e,n}\}$ \{ties are broken arbitrarily\} \label{alg:argmaxline}
		\State $r\gets0$
		\For{$t=t_n$ to $t_{n}+T_m-1$}
			\State Receive observation $y_t$
			\State Receive choice of action $a_t$ from expert $e_{t_n}$
			\State Apply $a_t$ to environment, receive reward $r_t$
			\State $r\gets r+r_t$
		\EndFor
		\State \textbf{end for}
		\State $S_{e_{t_n}} \gets S_{e_{t_n}}+ r/T_n$
		\State $k_{e_{t_n}} \gets k_{e_{t_n}}+1$
		\State $t_{n+1}=t_n+T_n$
 	\EndFor
	\end{algorithmic}
	\label{alg:ucb}
\end{algorithm}

\subsubsection{Construction of confidence bounds}

In the classical UCB algorithm, each `pull' of an expert gives an independent, and identically distributed sample from the expert's distribution. Thus the Azuma-Hoeffding inequality can be straight-forwardly applied to the empirical mean of the samples to construct an upper confidence bound. In the MDP setting, the cumulative reward is not independent of the past history of rewards because of the underlying dependency on the state of the MDP when the expert is chosen. To construct the confidence bounds, $c^e_{k,n}$ we first build a martingale out of the average cumulative rewards received for the expert from each time it is chosen, $R^e_k$. We then construct the confidence bound on the distance between the value of the martingale and the expected average cumulative reward of the expert under their steady-state distribution, $\bar R^e$.

We first construct a martingale from the sequence of rewards $(R_k^e)_k$. We begin by defining our filtration:\\

 Let $\Fc_k^e = \sigma(s_{t_1^e},\dots,s_{t_{k+1}^e},R_1^e,\dots,R_k^e)$, the $\sigma$-algebra generated by the sequence of initial states, including the next initialization, and average cumulative rewards.\footnote{We note that a quirk of the analysis is that the filtration at time $t_k^e$ requires the initial state for the expert $e$ at time $t_{k+1}^e$. Intuitively, this means that the filtration of our martingale construction has enough information to essentially decouple and ignore the actions of all the experts chosen in the intermediate time steps when we take the conditional expectation.} Define $Z_k^e = R_k^e - \E[R_k^e \vert \Fc_{k-1}^e]$. Finally, let $S_k^e = \sum_{\ell = 1}^k Z_\ell^e$. We claim that $S_k^e$ is a martingale. We can clearly see this, since $\E[S_{k+1}^e \vert \Fc_k^e] = \E[Z_{k+1}^e \vert \Fc_k^e] + \E[S_k^e \vert \Fc_k^e] = S_k^e$, and $\E[|S_k^e|] < \infty$ since our rewards are bounded. \\

Now, we have defined a martingale for each expert $e \in \Ec$. Furthermore, we note that each martingale has bounded differences, by the boundedness of our reward function: $|S_k^e - S_{k-1}^e| = |Z_k^e| \leq 1$ almost surely. With this construction we invoke the Azuma-Hoeffding inequality.

\begin{proposition}[Azuma-Hoeffding inequality~\citep{Hoeffding1963,Azuma1967}]
\label{prop:azuma}
For any expert $e \in \Ec$ and $t > 0$, we have:
\begin{equation}
\Pr\left( S_k^e \leq -t \right) \leq \exp \left( \frac{-t^2}{2k} \right)
\end{equation}
\end{proposition}

Given we have a high-probability bound on our martingale, we now bound the distance of the empirical mean of the average cumulative reward of the expert $e$, from its true value. Since $e$ is fixed, we will drop the dependence for cleanliness.

The event that we would like to have a high probability bound for is:

\[\left\{ \bar R - \frac{1}{k} \sum_{\ell = 1}^k R_\ell \geq t \right\} \]

Note that $\E[R_k^e \vert \Fc_{k-1}^e] = \E[R_k^e \vert s_{t_k^e}]$, by the Markov properties in the appropriate locations. We therefore derive a condition on our constructed martingale such that this event must hold:

\begin{align*}
\left\{ \bar R - \frac{1}{k} \sum_{\ell = 1}^k R_\ell \geq t \right\} &=
\left\{ \bar R - \frac{1}{k} \sum_{\ell = 1}^k \E[R_k \vert s_{t_k}] + \frac{1}{k} \sum_{\ell = 1}^k \E[R_k \vert s_{t_k}]  -  \frac{1}{k} \sum_{\ell = 1}^kR_\ell \geq t \right\} \\
&=\left\{ \bar R - \frac{1}{k} \sum_{\ell = 1}^k \E[R_k \vert s_{t_k}] - \frac{S_k}{k} \geq t \right\}
\end{align*}

This last event is a subset of an event whose probability we can bound:

\[
\left\{ \bar R - \frac{1}{k} \sum_{\ell = 1}^k \E[R_k \vert s_{t_k}] - \frac{S_k}{k} \geq t \right\} \subseteq
\left\{ \frac{1}{k} \sum_{\ell = 1}^k \frac{K}{T_0}  - \frac{S_k}{k} \geq t \right\} 
\]

\begin{align*}
\left\{ \frac{1}{k} \sum_{\ell = 1}^k \frac{K}{T_0}  - \frac{S_k}{k} \geq t \right\} &=
\left\{\frac{K}{T_0} - \frac{S_k}{k} \geq t \right\}\\
& =
\left\{\frac{S_k}{k} \leq \frac{K}{T_0} - t \right\}
\end{align*}
Thus:
\[
\Pr \left( \bar R - \frac{1}{k} \sum_{\ell = 1}^k R_\ell^k \geq t \right) \leq 
\Pr \left( \frac{S_k}{k} \leq \frac{K}{T_0} - t  \right)
\]
Now, let $t'= tk - k\frac{K}{T_0}$. By invoking Proposition~\ref{prop:azuma}, we get:

\[
\Pr \left( \bar R - \frac{1}{k} \sum_{\ell = 1}^k R_\ell \geq t \right) \leq \exp \left( \frac{-(t')^2}{2k} \right) = \exp \left( \frac{-k(t - \frac{K}{T_0})^2}{2} \right)
\]
Now for a fixed probability, $\delta>0$, we solve for $t$:
% \[
% \exp \left( \frac{-k(t - \frac{K}{T_0})^2}{2} \right) = \delta
% \]

% \[
% \log(\delta) = \frac{-k(t - \frac{K}{T_0})^2}{2}
% \]
% \[
% \sqrt{ -\frac{2}{k}\log(\delta)} = {t - \frac{K}{T_0}}
% \]
\[
t=\frac{K}{T_0} + \sqrt{ \frac{2}{k}\log\left(\frac{1}{\delta}\right)} 
\]
Thus, we have constructed a high probability bound on the event of interest:

\[
\Pr \left( \bar R^e \geq \frac{1}{k} \sum_{\ell = 1}^k R_\ell^e + \frac{K_e}{T_0} + \sqrt{ \frac{2}{k}\log\left(\frac{1}{\delta}\right)} \right) \leq \delta
\]

We first note that, in practice, we choose $\delta$ to vary with the number of iterations in our algorithm. In particular, we choose $\delta=n^{-4}$. Thus our confidence bound for each expert depends on both the number of times the expert was chosen, $k$, but also the number of iterations of the algorithm, $n$. Thus our final confidence bound is:

\[
c^e_{k,n}= \frac{K_e}{T_0}+\sqrt{\frac{8\log(n)}{k}}
\]

\subsubsection{A regret bound for the UCB algorithm in MDPs}

Having built the confidence bounds, $c^e_{k,n}$ for the UCB algorithm, we now use them to find an upper bound on the expected regret. Recall that in Definition \ref{def:regret} we have defined regret in this setting as: 

\[r(n)=n\bar R^* -  \sum_{k = 0}^{n}\sum_{t=t_k}^{t_{k+1}} \frac{r_t}{T_k}\]

Given Proposition \ref{prop:rdi}, all that is left to do is construct an upper bound on $\E[T_e(n)]$.\\

\begin{proposition}[Regret bound for the UCB algorithm in MDPs]
\label{prop:ucbregret}
For a set of experts $\Ec$ and initial time horizon $T_0$ such that $\bar{R^*}-\bar{R}^e > \frac{2K_e}{T_0}$ for all $e \in \Ec$, and for a choice of $\delta(n)=n^{-4}$, the expected regret of the UCB algorithm after $n$ time steps of the algorithm, $\E[r(n)]$, satisfies:

% \[
% \E[r(n)]\le \sum_{\substack{e \in \Ec\\ e\ne e^*}}{\bigg[ \bigg( \frac{8 \log{n}}{(\Delta_e -\frac{2K_e}{T_0})^2} +1+\frac{\pi^2}{3} \bigg) \bigg( \Delta_e+\frac{K_e}{T_0}}\bigg) \bigg]+K_*\log{\bigg(\frac{n+T_0}{1+T_0}\bigg)}
% \]

\[
\E[r(n)]\le \sum_{\substack{e \in \Ec\\ e\ne e^*}}{\bigg[ \bigg( \frac{32 \log{n}}{(\Delta_e -\frac{2K_e}{T_n})^2} +1+\frac{\pi^2}{3} \bigg) \bigg( \Delta_e+\frac{K_e}{T_0}}\bigg) \bigg]+K_*\sum_{k = 0}^{n-1} \frac{1}{T_k}
\]

\end{proposition}

We note that the term $K_*\sum_{k = 0}^{n-1} \frac{1}{T_k}$, is left as a design choice for now, since the sequence $(T_n)_n$ can be chosen. We would intuitively like $\{T_n, n\ge0\}$ to grow as slowly as possible to speed up our algorithm, but we also need the sum to grow as slowly as possible to control the rate of growth of the regret. We will ultimately make a choice of $\{T_n, n\ge0\}$ such that the sum does not negatively impact our asymptotic rate. 

\begin{proof}
We show that for $e \in \Ec, e\ne e^*$: 

\[
\E[T_e(n)]\le \frac{32\log{n}}{\left(\Delta_e -\frac{2K_e}{T_n}\right)^2} +1+\frac{\pi^2}{3} 
\]

To construct this upper bound, we follow the template of the proof of Theorem 1 from \cite{Auer2002}. They show that, for a chosen positive integer $w$:

\[
T_e(n)\le w+\sum_{m=1}^\infty \sum_{s=1}^{m-1}\sum_{k=w}^{m-1}\I\{\bar{R}^*+c^*_{s,m}\le \bar{R}_e+c^e_{k,m}\}
\]

The event $\{\bar{R}^*+c^*_{s,m}\le \bar{R}_e+c^e_{k,m}\}$ implies that at least one of the following must hold:
\[
R^*_s\le \bar{R}^*-c^*_{s,m}
\]
\[
R^e_k\ge \bar{R}^e+c^e_{k,m}
\]
\[
\bar{R}^*<\bar{R}^e+2c^e_{k,m}
\]

We note that by construction, $Pr(R^*_s\le \bar{R}^*-c^*_{s,m})\le m^{-4}$. Similarly, $Pr(R^e_k\ge \bar{R}^e+c^e_{k,m})\le m^{-4}$. We further choose $w$ such that the third inequality is always false:

\begin{align*}
\Delta_e & > 2c^e_{w,n}\\
\Delta_e - \frac{2K_e}{T_n} &> 2\sqrt{\frac{8\log{n}}{w}}\\
w&>\frac{32\log{n}}{\left(\Delta_e - \frac{2K_e}{T_n}\right)^2}\\
\end{align*}

Thus we get:
\begin{align*}
\E[T_e(n)] &\le \left\lceil\frac{32\log{n}}{\left(\Delta_e - \frac{2K_e}{T_n}\right)^2}\right\rceil+\sum_{m=1}^\infty \sum_{s=1}^{m-1}\sum_{k=w}^{m-1} 2m^{-4}\\
&\le  \left\lceil\frac{32\log{n}}{\left(\Delta_e - \frac{2K_e}{T_n}\right)^2}\right\rceil+ \sum_{m=1}^\infty \sum_{s=1}^{m-1}\sum_{k=1}^{m-1} 2m^{-4}\\
\end{align*}

Which we can reduce to:

\[
\E[T_e(n)]\le \frac{32\log{n}}{\left(\Delta_e - \frac{2K_e}{T_n}\right)^2}+1+\frac{\pi^2}{3}
\]

Using this upper bound gives us our desired result.
\end{proof}

We note that this analysis gives us a $O(\log{n})$ rate of convergence without looking at the term that depends on the sequence $\{T_n,n\ge0\}$. Since picking a sequence $\{T_n,n\ge0\}$ that grows too quickly will slow down our algorithm in practice, we can choose $T_n$ to grow linearly, i.e. $T_n=T_0+cn$ for some constant $c > 0$, and still preserve the logarithmic rate. Indeed, with this choice:

\[
\sum_{k = 0}^{n-1} \frac{1}{T_k}\le \frac{1}{c}\log{\left( \frac{T_0+cn-c}{T_0}\right)}
\]

Thus the expected regret under the UCB algorithm applied to MDPs is:

\[
\E[r(n)]\le \sum_{\substack{e \in \Ec\\ e\ne e^*}}  \left[ \left( \frac{32 \log{n}}{\left(\Delta_e - \frac{2K_e}{T_n}\right)^2} +1+\frac{\pi^2}{3} \right) \left( \Delta_e+\frac{K_e}{T_0} \right) \right]+\frac{K_*}{c}\log{\left( \frac{T_0+cn-c}{T_0}\right)}
\]

Our final result suggests that the expected regret grows at a rate of $O(\log{n})$ as $n \rightarrow \infty$. Further, we remark on the dependence on the initial time horizon $T_0$. As $T_0 \rightarrow \infty$ the regret reduces to the rate of the UCB algorithm under the analysis in \cite{Auer2002}, with the slight caveat that the factor of $32$ in our analysis is a factor of $8$ in theirs. This is due to the fact that we use a different concentration inequality since we are dealing with martingales. This matches our intuition that, as $T_0 \rightarrow \infty$, the individual choice of experts become less and less inter-dependent and begin to converge to the expected reward under the stationary distribution. Thus, in the limit, this problem reduces to the classical multi-armed bandit case.

It is worth noting that all the classical multi-armed bandit results will talk about rates in terms of $n$, i.e. the iteration number of the multi-armed bandit algorithm; for our application, what we often will care about in practice is the convergence rate in terms of $t$, i.e. the time step of the underlying system. Whereas our regret as a function of $n$ will only improve as $T_0 \rightarrow \infty$, we will find that the actual time it takes to converge in $t$ may group. This is an interesting trade-off that we are currently investigating.

\subsection{Comments}

In this paper, we have introduced algorithms which utilize a multi-armed bandit approach to select a policy from a finite set of candidate policies. Classical multi-armed bandit approaches usually assume that each time a bandit $e \in \Ec$ is used, the controller receives one independent and identically distributed sample from some fixed distribution. When the underlying system is an MDP, we do not have either of these conditions: the rewards received will depend on the state of the system determined by previous experts' actions, and the initialization will vary at each sample.

We have shown that by choosing a fixed expert for a long enough time horizon and averaging the rewards experienced and under some regularity conditions, the MDPs will sufficiently mix under a fixed policy. This causes each sample to be, in some sense, representative of the rewards of listening to expert $e$ in the long term.

\section{Experimental Results}
\label{sec:result}

% !TEX root = 2017-ArXiV.tex

We now seek to validate the theoretical results presented in the prequel on an example. We first outline the environment, the experts in the given set, and the choice of variables used in the experiments. We then outline preliminary results validating the logarithmic growth of the regret, as well as results showing the impact on the initial time horizon. Lastly, we show qualitative results on the ability of the UCB algorithm to respond to changes in the underlying MDP.

\subsection{Experimental Setup}

We test the UCB algorithm outlined in Section \ref{sec:problem} on a grid-world example. The goal of the ``agent" in the grid-world MDP is simply to maximize its cumulative reward. We first describe the dynamics, and then the reward structure of the grid.\\

 The grid and setup used is shown and described in Figure \ref{fig:setup}A. Each state, in this MDP, is a tile of the grid, and the agent begins in the white tile. The action space is $\Ac=\{1,2,3,4\}$. Each action corresponds to a move in one of the four cardinal directions. The dark blue states are almost-trapping states, where with probability 0.98 the agent is stuck in that state, and with probability 0.02, the agent moves in the direction of their chosen action. For all other states, the agent follows its desired action with probability 0.97, and goes in a random other direction with probability 0.03. In states along the edge of the grid, an action going out of the grid will result in a movement in a random direction back into the grid. 

\begin{figure}[!h]
 \centering
 \includegraphics[width=\textwidth]{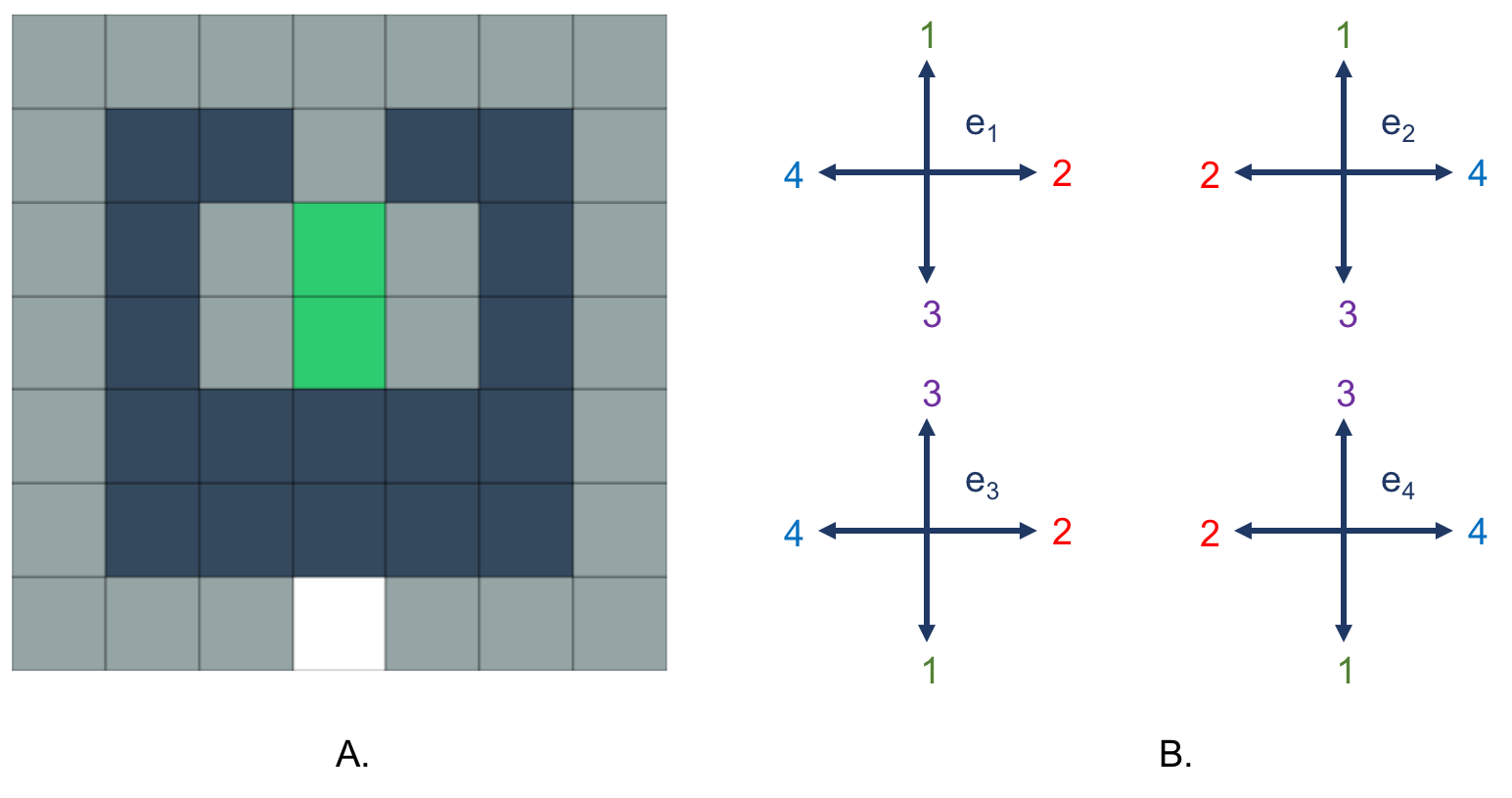}
 \caption{A. The setup of the Gridworld used for the experiments: Blue states are almost-trapping states, where with probability 0.98 you stay in the state, and  with probability 0.02 you move in the desired direction. In all other states, you move in the desired direction with probability 0.97 and a random direction with probability 0.03. The probability of pursuing an action going out of the grid is 0, and results in a movement in a random direction back into the grid. Green states give a reward of 1, dark blue states a reward of 0, and all other states a reward of 0.1. The initial state is the white tile. B. The dynamics under which each expert $e_1,e_2,e_3,e_4$ were trained.}
\label{fig:setup}
\end{figure}

The rewards in the grid are as follows: the green squares both give a reward of 1, while the gray squares give rewards of 0.1. Dark blue squares give a reward of 0. 

Given this setup, the optimal policy would have the agent move around the edge of the grid, avoiding the dark-blue states, and then go back and forth between the two green states.

The agent, in this problem, is supplied with a collection of expert policies $e_1,e_2,e_3,e_4$, and no knowledge of the underlying dynamics. Each expert policy was trained on an MDP with the same structure, but with the actions leading to movements in different directions. For example, for experts $e_1$ and $e_2$, the action $1$ corresponds to a movement `up' in the grid, while for experts $e_3$ and $e_4$, action $3$ corresponds to a movement `up' in the grid. The map from actions to movement that each expert was trained under is shown in Figure \ref{fig:setup}B.

Given this grid-world MDP and collection of experts $\Ec=\{e_1,e_2,e_3,e_4\}$, we now run the UCB algorithm with the confidence bounds and sequence $\{T_n,n\ge0\}$ used in the derivation of the regret bound in Section \ref{sec:problem}. For ease of reference we rewrite them here:

\[ T_n=T_0+cn\]
\[ c_{k,n}^e=\frac{K_e}{T_0}+\sqrt{\frac{8\log{n}}{k}}\]

We set $c=0.1$ for all the experiments and keep $T_0$ as a variable for now. Given that we have full information over the experts and the MDP in this case, we can also calculate $K_e$, for each expert: 

\[K_e=\frac{2}{1-\alpha_e}\]

Where $\alpha_e$ is the second largest eigenvalue of the probability transition kernel of the markov chain induced by the policy of expert $e$.

We note that the MDP is simple by construction, and we use it more for illustrative purposes than to show the algorithm at work on a difficult task. Further, we note that the confidence bounds and sequence of time horizons we use are those used in the derivation of the regret bound. There are choices that give better qualitative results, but for the sake of consistency we decided to stay with the same choices as in Section \ref{sec:problem}.

\subsection{Experimental Results}

Given the setup in the prequel, we now show the expected regret of the UCB algorithm. We run the algorithm 10 times for various values of $T_0$ and plot the average regret in Figure \ref{fig:regret}. The values of $K_e$, $\bar{R}^e$, and $\Delta_e$ for each expert in $\Ec$ are shown in Table \ref{fig:table}. We note that expert $e_1$ is the best expert in the set.

\begin{figure}[!ht]
\begin{center}

\begin{tabular}{|| c | c | c | c|| }
\hline
Expert 		& $K_e$ & $\bar{R}^e$ & $\Delta_e$\\
\hline 
$e_1$       & 2.02 & 0.74 & 0    \\
$e_2$       & 2.03 & 0.03 & 0.71 \\
$e_3$       & 2.00 & 0.08 & 0.66 \\
$e_4$       & 2.00 & 0.09 & 0.65 \\
\hline
\end{tabular}

\end{center}
\caption{A table outlining the values of the variables in the Section \ref{sec:problem} for the experts in $\Ec$.}
\label{fig:table}
\end{figure}

We can clearly see in Figure \ref{fig:regret}, the logarithmic rate of growth of the regret. Further, we can see how long time horizons leads to lower regret. This most likely occurs because the longer time horizon allows the average cumulative reward from each choice of expert to better approximate the average stead-state reward of that expert. We note, however, that in Figure \ref{fig:reward} we see that from the point of view of cumulative reward, the longer initial time horizon leads to a much slower rate of convergence to the optimal average cumulative reward. This mostly likely is due to the fact that listening to the wrong expert results in a higher corresponding loss of reward at each time step, than with smaller time horizons.

\begin{figure}[!h]
\centering
\begin{subfigure}[b]{0.49\textwidth}
  \centering
  \includegraphics[width=\textwidth]{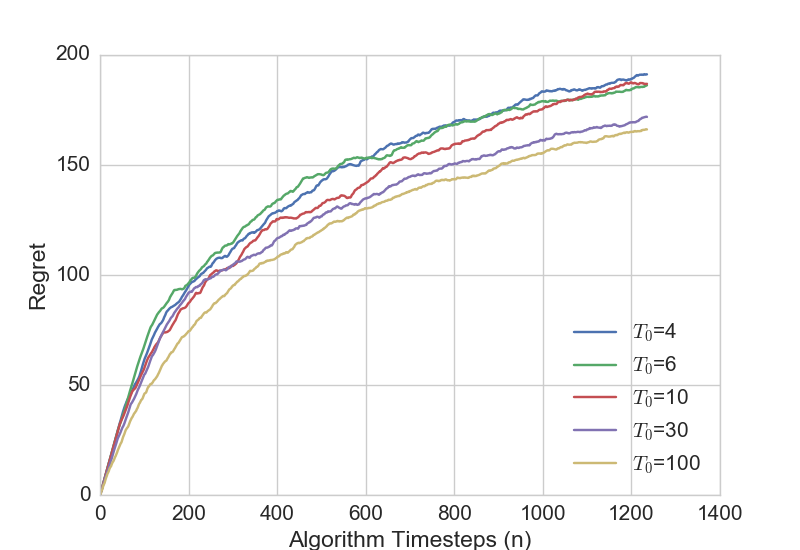}
  \caption{Average Regret averaged over 10 runs with various values of $T_0$}
  \label{fig:regret}
\end{subfigure}%
\hfill
\begin{subfigure}[b]{0.49\textwidth}
  \centering
  \includegraphics[width=\textwidth]{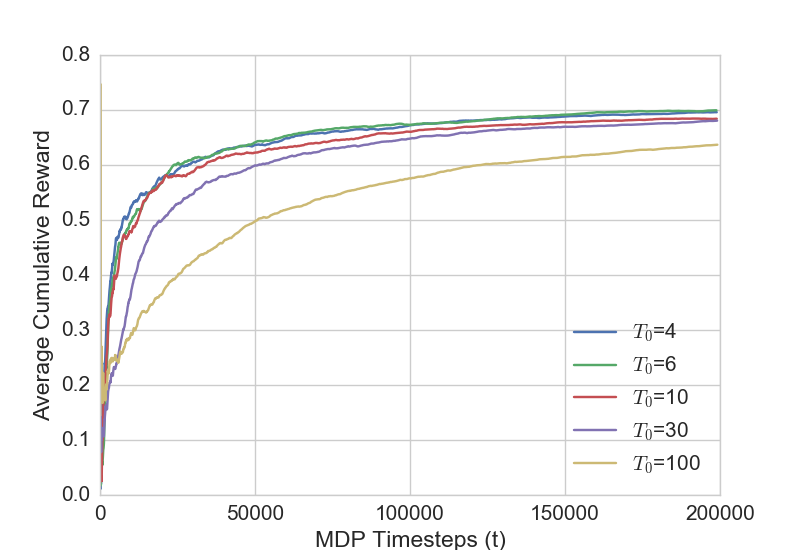}
  \caption{Average Cumulative Reward averaged over 10 runs with various values of $T_0$}
  \label{fig:reward}
\end{subfigure}
\caption{Results from running the UCB algorithm 10 times for various values of $T_0$ and then taking the average.}
\label{fig:results}
\end{figure}

We again note that the time-step of the algorithm is an artificial time-scale imposed over the problem to be able to formulate the problem as an MAB problem. Thus, the trend that a longer time horizon has lower regret after a fixed number of algorithm iterations is not necessarily a valid comparison, since correspondingly more time has elapsed in terms of the MDP. Indeed, given the results in Figure \ref{fig:results}, and with our stated goal of maximizing the cumulative reward, we would, perhaps counter-intuitively, choose the initial time horizon $T_0=4$ despite the fact that it incurs a larger regret in terms of the algorithm. We note however, that lower time horizons inherently give samples with higher variance, meaning that this trend may not hold true for all MDPs. Further, we note that choosing a larger $T_0$ is not asymptotically penalized since it does, eventually, achieve the same average cumulative reward as the lower $T_0$. Finally we note that all choices of $T_0$ have logarithmically growing regret and the average cumulative reward seems to converge, in all cases, to the average steady-state reward of the best expert in the set.

We show in Figure \ref{fig:perturb}, how the UCB algorithm responds to changes in the underlying MDP. At the 5000th iteration of the algorithm, the dynamics of the MDP are changed such that expert $e_1$ is no longer the best expert to listen to. The regret of the UCB algorithm in this situation is plotted in Figure \ref{fig:perturb}. We can see that the algorithm quickly adjusts to a new expert, and the regret continues to grow at the logarithmic rate. 

\begin{figure}[!h]
 \centering
 \includegraphics[width=0.8\textwidth]{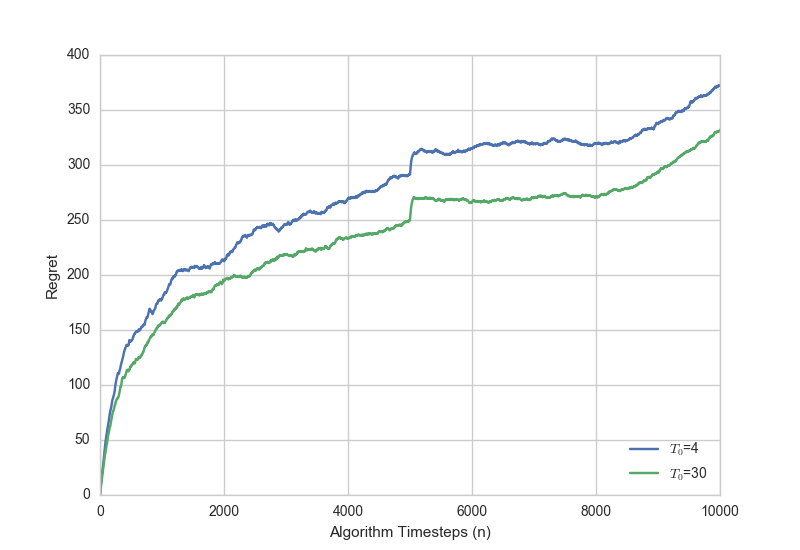}
 \caption{Regret of the UCB algorithm when there is a change in the dynamics of the underlying MDP.}
\label{fig:perturb}
\end{figure}

\section{Conclusion}
\label{sec:conclusion}
% !TEX root = 2017-ArXiV.tex

In this paper, we have outlined a general formulation for using a multi-armed bandit framework to choose an expert in Markov decision processes. We first outlined the problem of choosing the best expert out of a given set of experts, online, in an MDP. We then derived a basic decomposition for the regret of a multi-armed bandit algorithm applied to this problem. The regret decomposition was in the form of an upper bound on the expected regret of an algorithm. We defined regret as the difference between the cumulative average reward of the algorithm and the expected steady-state reward of the best expert in the set. We showed that under certain regularity conditions in the Markov chains induced by the policies of each expert we could upper bound the expected regret by a term depending linearly on the number of times a suboptimal expert has been chosen, and another term that depends only on the sequence of time horizons that the experts are listened to. We finally adapted the classical Upper Confidence Bounds algorithm to this new situation and showed that the regret was $O(\log(n))$. We then tested the performance of the algorithm on a GridWorld with unknown dynamics. We showed both the effects of varying the initial time horizon $T_0$ in the UCB algorithm, as well as the logarithmically growing regret. We further showed, qualitatively, how the UCB algorithm could rapidly respond to changes in the environment.

In future work, we hope to explore different multi-armed bandit algorithms that can have better performance in this setting. For example, MAB algorithms that take into account time-varying dynamics in the distributions of each slot machine may be more robust to changes in the underlying MDP. Further, we hope to apply this approach to real-life systems where maintaining good performance is critical. Finally, the regret bound derived for the UCB algorithm is loose, and a tighter bound could be more useful in practical situations.

%%%%%%%%%%%%%%%%%%%%%%%%%%%%%%%%%%%%%%%%%%%%%%%%%%%%%%%%%%%%%%%%%%%%%%%%%%%%%%%%
%\section{ACKNOWLEDGMENTS}
%
%The authors gratefully acknowledge the contribution of National Research Organization and reviewers' comments.
\appendix
\section{Proof of Corollary \ref{corr:mc_bnd}}
\label{sec:cor1}
% !TEX root = 2017arxivNew.tex

\begin{proof}

To prove Corollary \ref{corr:mc_bnd}, we first define the distribution over states resulting from following the stationary policy of expert $e$ for $i$ time steps and starting from state $s_{t_n}$.

We first define:
\[
\mu_{t_n}(s) = \I\{s = s_{t_n}\}
\]
This is the degenerate distribution that puts all its mass at the state at time $t_n$, which is given.

%\[
%\mu_{t_n+1}\left(s'\right)=\tilde{P}_e\left(s_{t_n},s'\right)
%\]

Then we can recursively define:

\[
\mu_{t+1}\left(s'\right)=\sum_{s \in \Sc} \tilde{P}_e\left(s,s'\right)\mu_{t}(s) \ \ \forall t = t_n,t_n+1,\dots,t_n+T_n-1
\]

Recall that $t_n+T_n = t_{n+1}$. Also, we take $e = e_{t_n}$. 
Having done this, we now expand and cleverly contract:

\[
\left| \bar R^e-\E\left[ \frac{1}{T_n}\sum_{t = t_n}^{t_{n+1}-1}r_t \middle| s_{t_n},e_{t_n}\right]\right|
\]

{\allowdisplaybreaks
\begin{align*}
&=\left|\frac{1}{T_n} \sum_{t = t_n}^{t_{n+1}-1}
	\left[
	\sum_{s'\in \Sc} \sum_{a\in \Ac}\sum_{s \in \Sc} \int_0^1r R(s,a,s',dr)P(s,a,s')\pi_e(s,a)\left[ \mu_e(s)-\mu_t(s) \right]
	\right]
\right| \\
&=\left|
\frac{1}{T_n} \sum_{t = t_n}^{t_{n+1}-1}
	\left[
	\sum_{s\in \Sc} \left[ \mu_e(s)-\mu_t(s) \right] \sum_{a\in \Ac} \sum_{s' \in \Sc} \int_0^1r R(s,a,s',dr)P(s,a,s')\pi_e(s,a)
	\right]
\right| \\
&\le\frac{1}{T_n} \sum_{t = t_n}^{t_{n+1}-1}
	\sum_{s\in \Sc} 
	\left|
	\left[ \mu_e(s)-\mu_t(s) \right] \sum_{a\in \Ac} \sum_{s' \in \Sc} \int_0^1r R(s,a,s',dr)P(s,a,s')\pi_e(s,a)
	\right|\\
&=\frac{1}{T_n} \sum_{t = t_n}^{t_{n+1}-1}
	\sum_{s\in \Sc} 
	\left| \mu_e(s)-\mu_t(s) \right|
	\left|
		\sum_{a\in \Ac} \sum_{s' \in \Sc} \int_0^1r R(s,a,s',dr)P(s,a,s')\pi_e(s,a)
	\right|\\
&\leq \frac{1}{T_n} \sum_{t = t_n}^{t_{n+1}-1}
	\sum_{s\in \Sc} 
	\left| \mu_e(s)-\mu_t(s) \right| =
\frac{1}{T_n} \sum_{t = t_n}^{t_{n+1}-1}
	\sum_{s\in \Sc} 
	\left| \mu_e(s)-\sum_{\tilde s \in \Sc} \tilde P_e^{t-t_n}(\tilde s,s)\mu_{t_n}(\tilde s) \right| \\
&=\frac{1}{T_n} \sum_{t = t_n}^{t_{n+1}-1}
	\sum_{s\in \Sc} 
	\left| \mu_e(s)-\tilde P_e^{t-t_n}(s_{t_n},s)\mu_{t_n}(s_{t_n}) \right|\\
&=\frac{1}{T_n} \sum_{t = t_n}^{t_{n+1}-1}
	\left\| \mu_e(s)-\tilde P_e^{t-t_n}(s_{t_n},s)\mu_{t_n}(s_{t_n}) \right\|_1\\
&\leq \frac{1}{T_n} \sum_{t = t_n}^{t_{n+1}-1}
	C_e \alpha_e^{t-t_n}\\
&= \frac{C_e (1 - \alpha_e^{T_n})}{T_n(1-\alpha_e)}
\end{align*}
}

%\[
%\frac{1}{T_n} \sum_{t = t_n}^{t_{n+1}-1}
%	\sum_{s\in \Sc} 
%	\left| \mu_e(s)-\mu_t(s) \right|
%	\sum_{a\in \Ac} \sum_{s' \in \Sc} 
%	\left|
%		\int_0^1r R(s,a,s',dr)P(s,a,s')\pi_e(s,a)
%	\right| \leq
%\]

% \[
% \frac{1}{T_n} \sum_{t = t_n}^{t_{n+1}-1}
% 	\sum_{s\in \Sc} 
% 	\left| \mu_e(s)-\mu_t(s) \right| =
% \frac{1}{T_n} \sum_{t = t_n}^{t_{n+1}-1}
% 	\sum_{s\in \Sc} 
% 	\left| \mu_e(s)-\sum_{\tilde s \in \Sc} \tilde P_e^{t-t_n}(\tilde s,s)\mu_{t_n}(\tilde s) \right| =
% \]
% \[
% \frac{1}{T_n} \sum_{t = t_n}^{t_{n+1}-1}
% 	\sum_{s\in \Sc} 
% 	\left| \mu_e(s)-\tilde P_e^{t-t_n}(s_{t_n},s)\mu_{t_n}(s_{t_n}) \right| =
% \frac{1}{T_n} \sum_{t = t_n}^{t_{n+1}-1}
% 	\left\| \mu_e(s)-\tilde P_e^{t-t_n}(s_{t_n},s)\mu_{t_n}(s_{t_n}) \right\|_1 \leq
% \]
% \[
% \frac{1}{T_n} \sum_{t = t_n}^{t_{n+1}-1}
% 	C_e \alpha_e^{t-t_n} = 
% \frac{C_e (1 - \alpha_e^{T_n})}{T_n(1-\alpha_e)}
% \]

The first equality is simply writing out the expectation in terms of the corresponding kernels, and follows by definition. The second line is a simple rearrangement by Fubini's theorem. The third line is a triangle inequality; the fourth is writing the modulus of a product in terms of the modulus of the individual terms. The fifth is due to the fact that the expected reward is always less than 1, even conditioned on $s$ and averaged across $a, s'$.

\end{proof}

%%%%%%%%%%%%%%%%%%%%%%%%%%%%%%%%%%%%%%%%%%%%%%%%%%%%%%%%%%%%%%%%%%%%%%%%%%%%%%%%

\bibliographystyle{plainnat}
\bibliography{DONG_ROY-refs,2017arxiv_refs}

\begin{thebibliography}{24}
\providecommand{\natexlab}[1]{#1}
\providecommand{\url}[1]{\texttt{#1}}
\expandafter\ifx\csname urlstyle\endcsname\relax
  \providecommand{\doi}[1]{doi: #1}\else
  \providecommand{\doi}{doi: \begingroup \urlstyle{rm}\Url}\fi

\bibitem[Auer et~al.(2002)Auer, Cesa-Bianchi, and Fischer]{Auer2002}
Peter Auer, Nicolo Cesa-Bianchi, and Paul Fischer.
\newblock Finite-time analysis of the multiarmed bandit problem.
\newblock \emph{Machine learning}, 47\penalty0 (2):\penalty0 235--256, 2002.

\bibitem[Azuma(1967)]{Azuma1967}
Kazuoki Azuma.
\newblock Weighted sums of certain dependent random variables.
\newblock \emph{Tohoku Math. J. (2)}, 19\penalty0 (3):\penalty0 357--367, 1967.
\newblock \doi{10.2748/tmj/1178243286}.

\bibitem[Bubeck and Cesa-Bianchi(2012)]{Bubeck2012}
Sébastien Bubeck and Nicolò Cesa-Bianchi.
\newblock Regret analysis of stochastic and nonstochastic multi-armed bandit
  problems.
\newblock \emph{Foundations and Trends in Machine Learning}, 5\penalty0
  (1):\penalty0 1--122, 2012.
\newblock ISSN 1935-8237.
\newblock \doi{10.1561/2200000024}.

\bibitem[Haeffele and Vidal(2015)]{Haeffele2015}
Benjamin~D. Haeffele and Rene Vidal.
\newblock Global optimality in tensor factorization, deep learning, and beyond.
\newblock \emph{arXiV}, 2015.

\bibitem[Hamrick et~al.(2017)Hamrick, Ballard, Pascanu, Vinyals, Heess, and
  Battaglia]{Hamrick2017}
Jessica~B. Hamrick, Andrew~J. Ballard, Razvan Pascanu, Oriol Vinyals, Nicolas
  Heess, and Peter~W. Battaglia.
\newblock Metacontrol for adaptive imagination-based optimization.
\newblock \emph{CoRR}, abs/1705.02670, 2017.
\newblock URL \url{http://arxiv.org/abs/1705.02670}.

\bibitem[Hoeffding(1963)]{Hoeffding1963}
Wassily Hoeffding.
\newblock Probability inequalities for sums of bounded random variables.
\newblock \emph{Journal of the American Statistical Association}, 58\penalty0
  (301):\penalty0 13--30, 1963.
\newblock \doi{10.1080/01621459.1963.10500830}.

\bibitem[Kallenberg(2002)]{Kallenberg2002}
Olav Kallenberg.
\newblock \emph{Foundations of Modern Probability}.
\newblock Springer-Verlag New York, 2nd edition, 2002.
\newblock ISBN 978-0-387-95313-7.

\bibitem[Lai and Robbins(1985)]{lai1985}
Tze~Leung Lai and Herbert Robbins.
\newblock Asymptotically efficient adaptive allocation rules.
\newblock \emph{Advances in applied mathematics}, 6\penalty0 (1):\penalty0
  4--22, 1985.

\bibitem[Laskey et~al.(2015)Laskey, Mahler, McCarthy, Pokorny, Patil, Van
  Den~Berg, Kragic, Abbeel, and Goldberg]{Laskey2015}
Michael Laskey, Jeff Mahler, Zoe McCarthy, Florian~T Pokorny, Sachin Patil, Jur
  Van Den~Berg, Danica Kragic, Pieter Abbeel, and Ken Goldberg.
\newblock Multi-armed bandit models for {2D} grasp planning with uncertainty.
\newblock In \emph{Automation Science and Engineering (CASE), 2015 IEEE
  International Conference on}, pages 572--579. IEEE, 2015.

\bibitem[Levine et~al.(2016)Levine, Finn, Darrell, and Abbeel]{levine2016end}
Sergey Levine, Chelsea Finn, Trevor Darrell, and Pieter Abbeel.
\newblock End-to-end training of deep visuomotor policies.
\newblock \emph{Journal of Machine Learning Research}, 17\penalty0
  (39):\penalty0 1--40, 2016.

\bibitem[Matikainen et~al.(2013)Matikainen, Furlong, Sukthankar, and
  Hebert]{Matikainen}
Pyry Matikainen, P.~Michael Furlong, Rahul Sukthankar, and Martial Hebert.
\newblock Multi-armed recommendation bandits for selecting state machine
  policies for robotic systems.
\newblock In \emph{Proceedings of International Conference on Robotics and
  Automation (ICRA 2013)}, 2013.

\bibitem[Mnih et~al.(2013)Mnih, Kavukcuoglu, Silver, Graves, Antonoglou,
  Wierstra, and Riedmiller]{Minh_2013}
Volodymyr Mnih, Koray Kavukcuoglu, David Silver, Alex Graves, Ioannis
  Antonoglou, Daan Wierstra, and Martin~A. Riedmiller.
\newblock Playing {Atari} with deep reinforcement learning.
\newblock \emph{CoRR}, abs/1312.5602, 2013.
\newblock URL \url{http://arxiv.org/abs/1312.5602}.

\bibitem[Mnih et~al.(2016)Mnih, Badia, Mirza, Graves, Lillicrap, Harley,
  Silver, and Kavukcuoglu]{Mnih2016}
Volodymyr Mnih, Adri{\`{a}}~Puigdom{\`{e}}nech Badia, Mehdi Mirza, Alex Graves,
  Timothy~P. Lillicrap, Tim Harley, David Silver, and Koray Kavukcuoglu.
\newblock {Asynchronous Methods for Deep Reinforcement Learning}.
\newblock feb 2016.
\newblock URL \url{http://arxiv.org/abs/1602.01783}.

\bibitem[Ortner et~al.(2012)Ortner, Ryabko, Auer, and Munos]{Ortner2012}
Ronald Ortner, Daniil Ryabko, Peter Auer, and R{\'{e}}mi Munos.
\newblock Regret bounds for restless markov bandits.
\newblock \emph{CoRR}, abs/1209.2693, 2012.
\newblock URL \url{http://arxiv.org/abs/1209.2693}.

\bibitem[Osband et~al.(2016)Osband, Blundell, Pritzel, and Roy]{Osband2016}
Ian Osband, Charles Blundell, Alexander Pritzel, and Benjamin~Van Roy.
\newblock Deep exploration via bootstrapped {DQN}.
\newblock \emph{CoRR}, abs/1602.04621, 2016.
\newblock URL \url{http://arxiv.org/abs/1602.04621}.

\bibitem[Raghu et~al.(2017)Raghu, Poole, Kleinberg, Ganguli, and
  Sohl-Dickstein]{Raghu2017}
Maithra Raghu, Ben Poole, Jon Kleinberg, Surya Ganguli, and Jascha
  Sohl-Dickstein.
\newblock On the expressive power of deep neural networks.
\newblock \emph{arXiV}, 2017.

\bibitem[Robbins(1952)]{Robbins1952}
Herbert Robbins.
\newblock Some aspects of the sequential design of experiments.
\newblock \emph{Bull. Amer. Math. Soc.}, 58\penalty0 (5):\penalty0 527--535, 09
  1952.
\newblock URL \url{http://projecteuclid.org/euclid.bams/1183517370}.

\bibitem[Rothschild(1974)]{rothschild1974two}
Michael Rothschild.
\newblock A two-armed bandit theory of market pricing.
\newblock \emph{Journal of Economic Theory}, 9\penalty0 (2):\penalty0 185--202,
  1974.

\bibitem[Schulman et~al.(2015{\natexlab{a}})Schulman, Levine, Moritz, Jordan,
  and Abbeel]{Schulman_15_1}
John Schulman, Sergey Levine, Philipp Moritz, Michael~I. Jordan, and Pieter
  Abbeel.
\newblock Trust region policy optimization.
\newblock \emph{CoRR}, abs/1502.05477, 2015{\natexlab{a}}.
\newblock URL \url{http://arxiv.org/abs/1502.05477}.

\bibitem[Schulman et~al.(2015{\natexlab{b}})Schulman, Moritz, Levine, Jordan,
  and Abbeel]{Schulman_15_2}
John Schulman, Philipp Moritz, Sergey Levine, Michael~I. Jordan, and Pieter
  Abbeel.
\newblock High-dimensional continuous control using generalized advantage
  estimation.
\newblock \emph{CoRR}, abs/1506.02438, 2015{\natexlab{b}}.
\newblock URL \url{http://arxiv.org/abs/1506.02438}.

\bibitem[Simon(1989)]{simon1989optimal}
Richard Simon.
\newblock Optimal two-stage designs for phase ii clinical trials.
\newblock \emph{Controlled clinical trials}, 10\penalty0 (1):\penalty0 1--10,
  1989.

\bibitem[St-Pierre et~al.(2011)St-Pierre, Louveaux, and Teytaud]{st2011online}
David~Lupien St-Pierre, Quentin Louveaux, and Olivier Teytaud.
\newblock Online sparse bandit for card games.
\newblock In \emph{ACG}, pages 295--305. Springer, 2011.

\bibitem[Tekin and Liu(2012)]{Tekin2012}
C.~Tekin and M.~Liu.
\newblock Online learning of rested and restless bandits.
\newblock \emph{IEEE Transactions on Information Theory}, 58\penalty0
  (8):\penalty0 5588--5611, Aug 2012.
\newblock ISSN 0018-9448.
\newblock \doi{10.1109/TIT.2012.2198613}.

\bibitem[Zhang et~al.(2017)Zhang, Bengio, Hardt, Recht, and Vinyals]{Zhang2017}
Chiyuan Zhang, Samy Bengio, Moritz Hardt, Benjamin Recht, and Oriol Vinyals.
\newblock Understanding deep learning requires rethinking generalization.
\newblock \emph{arXiV}, 2017.

\end{thebibliography}

\end{document}